\pdfoutput=1
\documentclass[10pt, conference, letterpaper]{IEEEtran}
\IEEEoverridecommandlockouts
\usepackage{cite}
\usepackage{amsmath,amssymb,amsfonts}
\usepackage{algorithmic}
\usepackage{graphicx}
\usepackage{textcomp}
\usepackage{xcolor}
\usepackage{color}
\usepackage{caption}

\usepackage{booktabs}
\usepackage{bm}

\usepackage{graphicx}
\graphicspath{ {Images/} }  

\usepackage[caption=false,font=footnotesize]{subfig}

\usepackage{color}
\usepackage{subfig}

\usepackage{algorithm}
\usepackage{algorithmic}   
\usepackage{comment}

\newtheorem{lemma}{\textbf{Lemma}}[section]
\newtheorem{theorem}{\textbf{Theorem}}[section]

\newtheorem{proposition}{\textbf{Proposition}}[section]

\def\BibTeX{{\rm B\kern-.05em{\sc i\kern-.025em b}\kern-.08em
    T\kern-.1667em\lower.7ex\hbox{E}\kern-.125emX}}
\begin{document}

\title{LYRA: Label-Free Structural Synchronization and Resource Allocation for UAV Edge Networks
}

\author{
    \IEEEauthorblockN{
        Feng He\IEEEauthorrefmark{1}\IEEEauthorrefmark{2},
        Alireza Furutanpey\IEEEauthorrefmark{2},
        Paolo Bellavista\IEEEauthorrefmark{1}, 
        Yu Qiu\IEEEauthorrefmark{3},
        Jiangchuan Liu\IEEEauthorrefmark{4}, 
        Jiannong Cao\IEEEauthorrefmark{5},
        Schahram Dustdar\IEEEauthorrefmark{2}
    }
    
    \vspace{1.5mm} 
    
    \IEEEauthorblockA{
        \IEEEauthorrefmark{1}University of Bologna,
        \IEEEauthorrefmark{2}TU Wien, 
        \IEEEauthorrefmark{3}City University of Hong Kong, \\
        \IEEEauthorrefmark{4}Simon Fraser University, 
        \IEEEauthorrefmark{5}Hong Kong Polytechnic University, \\
        Email: feng.he2@unibo.it, a.furutanpey@dsg.tuwien.ac.at, paolo.bellavista@unibo.it, csqiuyu@mail.scut.edu.cn, \\ jcliu@sfu.ca, jiannong.cao@polyu.edu.hk, dustdar@dsg.tuwien.ac.at
    }
}

\maketitle

\begin{abstract}
While deploying hierarchical vision models to process mission-critical tasks, UAV edge systems must adaptively update the models to sustain inference reliability under low-level environmental corruption. However, existing work has overlooked the optimal timing for model updates, the impracticality of relying on real-time expert labels, and the significant bandwidth and energy constraints of UAVs. This paper proposes a joint model update scheduling and resource allocation framework, aiming to maximize long-term semantic fidelity and resource efficiency of UAV edge intelligence systems. To address the challenge of label-free semantic evaluation, we formulate the Online Semantic Disagreement Rate (OSDR) as a proxy for timely update triggering, thereby enabling fine-grained Sensitivity-Aware Structural Synchronization (SASS). Furthermore, to overcome the curse of dimensionality in hybrid action spaces and effectively bound long-term energy budgets, we propose a Lyapunov-guided discrete reinforcement learning algorithm that performs action space dimensionality reduction and transforms constraints into virtual queue stability problems. The reported experimental results, based on real traffic traces, demonstrate that the proposed framework consistently outperforms representative baselines in semantic recovery efficiency and update triggering precision, by satisfying long-term energy budget and by reducing average risk backlog by up to 33.3\% in the dynamic environmental corruption scenario.
\end{abstract}



\section{Introduction}
MEC-enabled unmanned aerial vehicle (UAV) systems have emerged as a pivotal paradigm for mission-critical visual tasks, ranging from autonomous infrastructure inspection to real-time disaster response. To maintain high-level perception, these UAV edge nodes typically deploy sophisticated vision models characterized by hierarchical architectures \cite{r1,r2,r3}. However, the operational reliability of these systems is frequently compromised by low-level environmental corruptions, such as fog or abrupt illumination variations, encountered during continuous flight. These environmental stressors primarily corrupt low-level physical pixels, which, due to the feed-forward nature of hierarchical deep neural networks (DNNs), can trigger a significant degradation of semantic fidelity. Consequently, ensuring the continuous reliability of edge intelligence against such unpredictable stressors, while effectively bounding the system operational energy and bandwidth budgets, constitutes a primary challenge for resilient autonomous inspection. 

Despite the critical need for timely model maintenance, existing paradigms struggle to reconcile semantic reliability with resource efficiency under dynamic stressors. On the one hand, current frameworks predominantly rely on communication-intensive and full model synchronization \cite{r4,r5,r6,r7,r8,r9,r10,r11,r12,r13,r14,r15,r16,r17,r18,r19}, which incurs substantial overhead on the limited bandwidth of UAV-to-ground links. Moreover, these conventional trigger mechanisms are ill-suited for tracking inference degradation caused by dynamic environmental stressors. 
In fact, existing trigger mechanisms, whether time-driven (e.g., AoI, periodic schedules) or state-driven (e.g., resource utilization, data accumulation) \cite{r4,r5,r8,r9,r11,r12,r13,r14,r15,r16,r17,r19,r20,r21,r22,r23,r24,r25,r26}, are content-agnostic, relying solely on elapsed time or system states without perceiving the non-linear semantic degradation caused by abrupt visual corruptions.
On the other hand, recent semantic-aware scheduling approaches attempt to incorporate task-level feedback, but typically assume the real-time availability of ground-truth labels to compute inference accuracy or loss as feedback \cite{r6,r7,r10,r18}. In autonomous inspection flights, streaming visual data cannot be annotated by ground experts in real time, making label-dependent feedback physically unrealizable mid-flight, thus generating a relevant operational issue: the UAV cannot afford bandwidth-heavy full updates over air-to-ground links, yet lacks a label-free proxy to evaluate semantic degradation.

To address these open challenges, we propose LYRA, which, to the best of our knowledge, is the first joint update scheduling and resource allocation framework tailored for UAV edge inference. LYRA specifically optimizes synchronization efficiency under low-level environmental corruptions to maintain the semantic fidelity and resource efficiency of hierarchical vision models. To achieve these conflicting objectives, we introduce three interconnected innovations. First, we develop SASS, transitioning from communication-intensive full-model synchronization to fine-grained and partial updates. By adaptively regulating synchronization depth via dynamic sensitivity thresholds, SASS achieves bandwidth-efficient maintenance tailored to DNN hierarchical vulnerability. Second, to circumvent the unavailability of expert feedback, we formulate OSDR as a label-free proxy to quantify real-time semantic fidelity. Crucially, simply combining these mechanisms introduces non-trivial optimization challenges: dynamically coupling proxy-driven update triggers with SASS structural thresholds and continuous resources inherently creates an intractable hybrid continuous-discrete optimization problem. To resolve this integration complexity, we develop a Lyapunov-guided DRL that mathematically unifies these components, by performing action space dimensionality reduction and by decoupling long-term physical constraints to facilitate efficient exploration within the complex action space. Therefore, the main original contributions of the paper can be summarized as follows:

$\bullet$ \textbf{Novel Hierarchical Synchronization Mechanism}: We introduce SASS, which replaces conventional and communication-intensive full-model updates. By leveraging dynamic sensitivity thresholds to regulate synchronization depth, SASS executes fine-grained partial structural updates, by significantly reducing bandwidth consumption while accelerating semantic recovery.

$\bullet$ \textbf{Label-Free Semantic Performance Proxy}: We formulate OSDR to dynamically evaluate the Value of Update (VoU). This metric  circumvents the reliance on real-time ground-truth labels, by providing a physically realizable label-free proxy to quantify inference degradation under dynamic environmental corruptions.

$\bullet$ \textbf{Joint Optimization and Algorithm Design}: We formulate the joint model update scheduling and resource allocation problem as a highly non-convex Mixed-Integer Non-Linear Programming (MINLP) model. To solve this, we propose the LYRA framework, which originally integrates Lyapunov optimization with DRL. This approach decouples long-term hard physical constraints and performs action space dimensionality reduction, thereby mitigating the curse of dimensionality and satisfying long-term energy and resource compliance.

$\bullet$ \textbf{Trace-driven system evaluation}: Extensive experiments have been conducted to validate LYRA effectiveness. The reported quantitative results demonstrate that LYRA reduces average semantic risk backlog by 33.3\% over state-of-the-art baselines. Furthermore, SASS reduces average communication overhead by 80.5\% compared with conventional back-to-front updates, while achieving superior semantic recovery efficiency.

\section{Related Work}

\textbf{Edge AI Models and Continuous Maintenance.} Recent advancements in MEC have catalyzed the development of continuous edge learning and cloud-edge collaborative model maintenance to support dynamic inference tasks. However, predominant model update mechanisms typically rely on either full-model fine-tuning or back-to-front partial updates, i.e., retraining only deep semantic layers or classifier heads \cite{r27,r28}. These conventional paradigms are tailored for addressing label shifts or task variations. When subjected to low-level environmental corruptions, such as dynamic fog or illumination drops commonly encountered by UAVs, these strategies exhibit significant inefficiency. Full-model synchronization is impractical over bandwidth-constrained air-to-ground links. Furthermore, back-to-front updates are inherently information-limited when faced with significant low-level corruptions, as deep-layer refinements cannot recover information discarded by degraded shallow representations. This motivates SASS, which prioritizes front-to-back synchronization to better align with asymmetric sensitivity of hierarchical neural blocks.

\textbf{Model Update Triggers and Scheduling Metrics.} To optimize resource allocation in MEC networks, conventional update triggers generally fall into three categories: time-driven (e.g., periodic triggers \cite{r4,r5,r20,r24,r25} or AoI \cite{r8,r11,r14,r20,r21,r22,r23}) to minimize model staleness, internal system state-driven heuristics (e.g., resource utilization \cite{r12,r15,r16,r17,r19}, updated DT data accumulation \cite{r9,r13,r14,r15,r17,r19,r26}) to control system overhead, or model-driven metrics (e.g., loss or accuracy \cite{r6,r7,r10,r18}) to maximize model fidelity. Concurrently, emerging semantic-aware scheduling frameworks \cite{r28,r29,r30,r31,r32,r33} attempt to prioritize transmissions based on downstream task utility rather than mere bit-level metrics. Nevertheless, both state- and time-driven metrics are content-agnostic; they track elapsed time or system states, failing to capture the non-linear inference degradation caused by these environmental corruptions. Furthermore, SOTA semantic scheduling mechanisms often operate under the assumption of the real-time availability of ground-truth labels at the edge node to compute inference accuracy or loss as feedback. In realistic autonomous UAV operations, real-time visual streams are unlabelled, rendering such label-dependent feedback loops impractical to deploy online. To solve this challenge, we formulate the OSDR. By addressing the limitations of state/time-driven mechanisms and the label-dependence of model-driven approaches, this proxy drives the VoU evaluation.

\textbf{DRL for Joint Scheduling and Model Updates.} DRL have been extensively adopted to tackle dynamic resource allocation and trajectory planning in UAV-assisted MEC networks \cite{r34,r35,r36,r37,r38,r39,r40,r41,r42}. However, directly applying conventional DRL to our joint update scheduling and resource allocation problem encounters specific optimization challenges. The formulation is a highly non-convex MINLP problem featuring a hybrid action space: discrete update decisions coupled with continuous layer sensitive threshold, bandwidth and power allocations. To handle such hybrid spaces, existing DRL frameworks typically resort to either coarse discretization of continuous variables, leading to a curse of dimensionality and precision loss, or utilizing soft reward penalties that lack theoretical guarantees for long-term physical constraints, i.e., energy budgets. To address this convergence instability, our proposed Lyapunov-guided hierarchical architecture, which decouples long-term hard constraints and resolves the complex MINLP by mapping coupled hybrid decision variables into a unified discrete structural action space, paired with exact closed-form physical execution.


\section{Empirical Observations and Motivations}

This section presents a series of trace-driven empirical observations using ResNet-18 and CIFAR-10-C in order to reveal the physical failure mechanisms of edge vision models under non-stationary environments and to provide solid physical justifications for both our mathematical formulation of Section IV and our original solution for action space dimensionality reduction of Section VI.

\subsection{Vulnerability of edge vision model to low-level corruptions}

UAVs serving as MEC edge nodes frequently encounter highly dynamic environments, i.e., sudden fog or illumination drops, during continuous flight, inducing severe visual corruptions that directly degrade vision model reliability.


To quantify this vulnerability, we expose a pre-trained baseline model to varying severities of fog and contrast corruptions. As depicted in Fig. \ref{empirical_observations}(a), inference accuracy decreases non-linearly under environmental stress: fog corruption significantly destroys high-frequency spatial features, causing an immediate degradation (from 92.1\% to 47.4\%) at Severity 1, whereas contrast corruption remains initially stable before crashing below 20\% at Severity 5.

\textbf{Insight 1:} This unpredictable, non-linear degradation invalidates static deployments and periodic, AoI-based updates, necessitating a dynamic, on-demand update trigger ($\alpha_t \in \{0, 1\}$) tailored to real-time stress.

\subsection{Asymmetric layer sensitivity and trade-off}

Given strict UAV energy and bandwidth constraints, continuous full-model updates are impractical. To explore the optimal balance between communication cost and accuracy recovery, we investigate SASS.


Unlike traditional back-to-front fine-tuning for task shifts, these environmental corruptions primarily corrupt low-level physical pixels. If shallow receptive fields fail to extract reliable spatial features, deep semantic layers cannot recover the lost information, a phenomenon supported by Data Processing Inequality $I(Y;\hat{Y}) \le I(Y;H_k)$. Therefore, front-to-back synchronization is generally more resource-efficient than back-to-front updates under significant low-level physical corruptions. So, we propose the Front-to-Back Cumulative Update strategy. This mechanism inherently demands a sensitivity-aware threshold $\tau$ (see Section IV.B), where only the neural blocks most vulnerable to current environmental shift are selected for synchronization. Fig. \ref{empirical_observations} (b) reveals two critical physical properties of DNNs: (1) Diminishing marginal returns: Updating only the shallowest block (Conv1+Blk1) requires a minimal static patch size (0.574 MB) yet generates significant accuracy recovery (25.14\%). Conversely, extending updates to the deepest layers hits an accuracy ceiling (65.78\%) but incurs an exponential physical patch size explosion (42.662 MB). (2) Block-wise physical topology: CNN architectures inherently introduce discrete boundaries for feature transmission.

\textbf{Insight 2}: Block-wise structure and tradeoff curve legitimize continuous physical patch-size functions in Section IV and inspire our Action Space Dimensionality Reduction (Section VI), which maps complex continuous optimization into discrete macro-block actions to guarantee DRL convergence.

\begin{figure}[tbp]
    \centering
    \subfloat[Accuracy degradation]{\includegraphics[width=0.32\columnwidth]{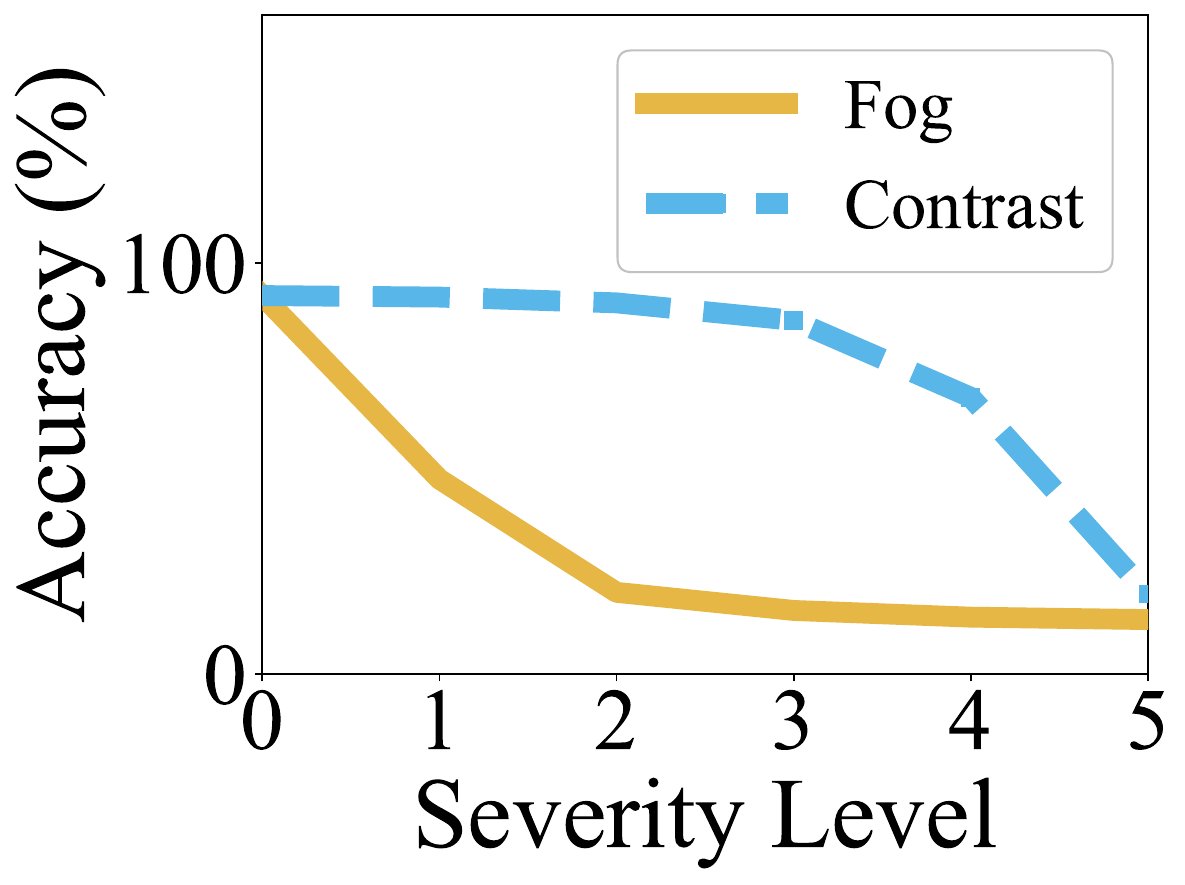}}
    \hfil 
    \subfloat[Cumulative patch-size trade-off]{\includegraphics[width=0.32\columnwidth]{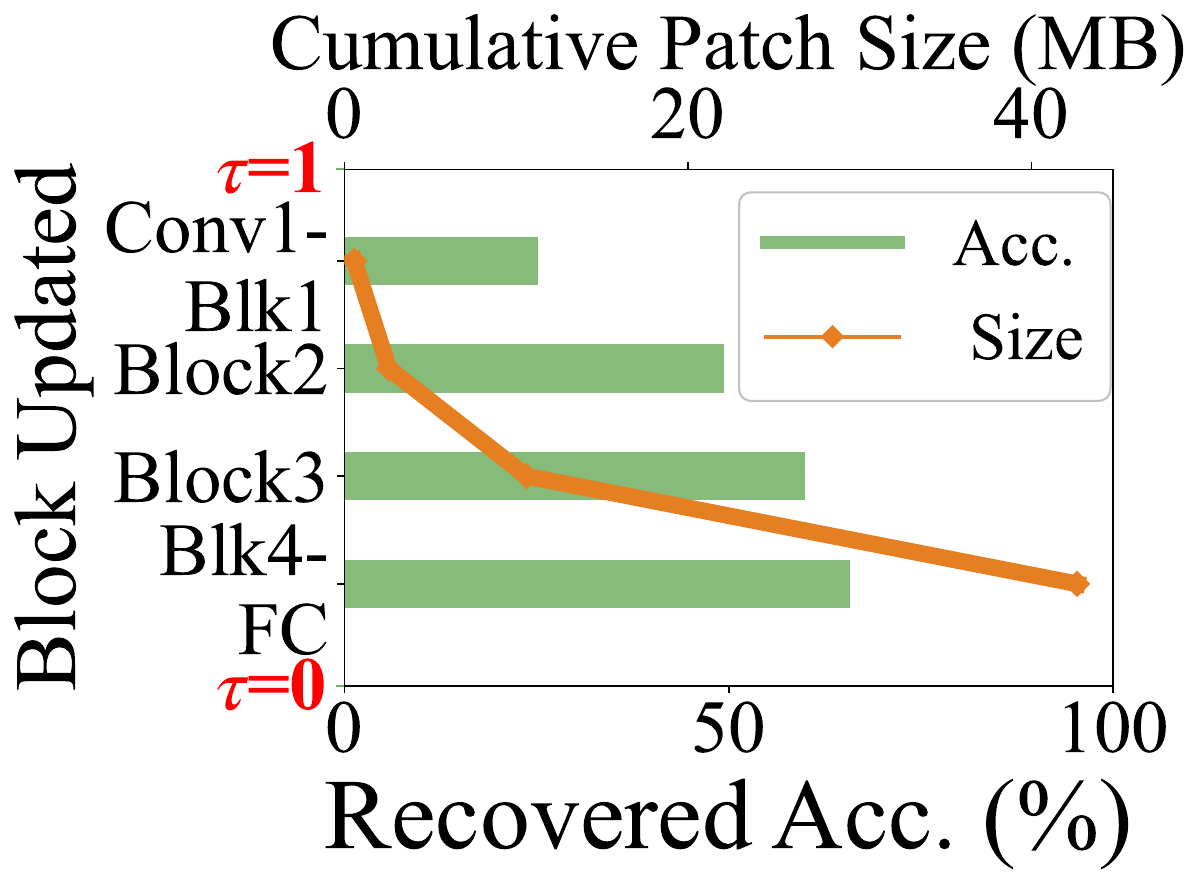}}
    \hfil
    \subfloat[OSDR correlation]{\includegraphics[width=0.32\columnwidth]{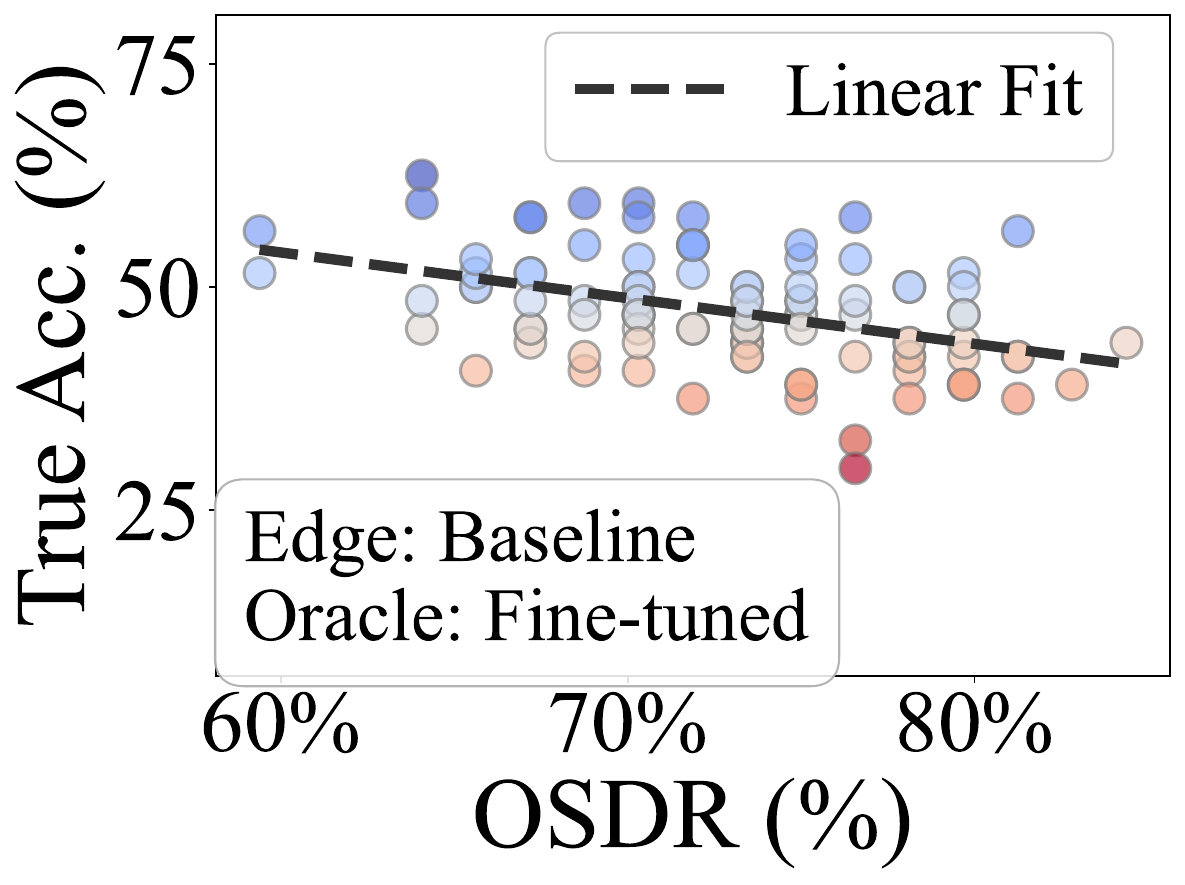}}
    
    \caption{
        Empirical observations on ResNet-18 and CIFAR-10-C: (a) non-linear accuracy degradation under corruption severity; (b) trade-off between recovered accuracy and cumulative patch size; (c) correlation between OSDR and true accuracy drop.}
    \label{empirical_observations}
\end{figure}

\subsection{Semantic Divergence as an Online Performance Proxy}

In real-world UAV MEC applications, ground-truth labels for real-time visual data are absent, rendering direct inference accuracy calculation impossible. To dynamically trigger SASS, we shift the optimization paradigm from AoI to VoU by defining a label-free proxy.

Specifically, we introduce OSDR to measure the prediction divergence between the deployed edge model and a server-hosted oracle on identical unlabelled batches. Practically, the oracle is updated offline periodically (e.g., weekly), remaining quasi-static during flights. At each slot, the UAV offloads a lightweight probe batch (32 compressed frames, 64 KB uplink overhead) to calculate the OSDR via the server-hosted oracle. Given ground server's parallel computing capability, evaluating a batch size of 32 maximizes GPU utilization under a negligible inference latency, while its transmission overhead over the 2 MHz link remains millisecond level. This fixed systemic overhead is accounted for in base budgets; although the oracle might be temporally stale, the generalization bound formally incorporates its error. As seen in Fig. \ref{empirical_observations}(c), evaluating OSDR against true accuracies reveals a strong negative correlation. While empirical linear fitting cannot exhaustively capture all unpredictable visual corruptions, it provides a crucial motivation: minimizing disagreement intrinsically restricts edge errors, establishing OSDR as a valid proxy.

\begin{theorem}
    Conditioned on the historical filtration $\mathcal{F}_{t-1}$. Let $e_t = \text{Pr}(f_t(X) \neq Y)$ be the true edge error and $d_t = \text{Pr}(f_t(X) \neq g_t(X))$ be the true edge-oracle disagreement rate under drifted distribution $P_t$. If oracle error satisfies $\text{Pr}(g_t(X) \neq Y) \le \epsilon_t$, the empirical OSDR $\hat{d}_t$ estimated over $m=32$ conditionally i.i.d. probe samples guarantees: $\text{Pr}\left( \vert{}e_t - \hat{d}_t\vert{} \le \epsilon_t + \sqrt{\log(2/\delta)/2m} \;|\mathcal{F}_{t-1} \right) \ge 1 - \delta$.
    \label{Semantic Consistency of OSDR}
\end{theorem}

\textbf{\textit{Proof}}: Pointwise error decomposition yields $|1\{f_t(X) \neq Y\} - 1\{f_t(X) \neq g_t(X)\}| \le1\{g_t(X) \neq Y\}$. Taking expectations over $P_t$ implies the population bound $\vert{}e_t - d_t\vert{} \le \epsilon_t$. Given $\mathcal{F}_{t-1}$, because deployed edge model $f_t$ and server oracle $g_t$ are deterministic, applying Hoeffding's inequality to the empirical average $\hat{d}_t$ over $m$ samples yields $\text{Pr}\left(\vert{}\hat{d}_t - d_t\vert{} > \sqrt{\log(2/\delta)/2m} \;|\mathcal{F}_{t-1}\right) \le \delta$. Combining these bounds via the triangle inequality and union bound establishes the obtained conclusion.


\begin{proposition}
    Let $\pi_{\text{label}}$ and $\pi_{\text{OSDR}}$ denote the idealized label-aware policy and the proposed empirical OSDR-driven policy, respectively. Under the conditions of \textit{Theorem} 3.1, the time-averaged regret $R(T) = \frac{1}{T}\mathbb{E}\left[\sum_{t=0}^{T-1}(\text{Cost}(\pi_{\text{OSDR}};t)-\text{Cost}(\pi_{\text{label}};t))\right]$ satisfies: $\lim_{T\rightarrow\infty} R(T) \le W_1 \cdot \left( \epsilon_{\text{max}} + \sqrt{\log(2/\delta)/2m} \right)$ where $W_1 > 0$ is a Lipschitz constant of the per-slot cost w.r.t. the trigger metric and $\epsilon_{\text{max}} = \sup_t \epsilon_t$. The asymptotic performance gap is thus strictly bounded by the finite-sample tracking error of the empirical proxy. 
    
    \label{Asymptotic Optimality of OSDR-Driven Policy}
\end{proposition}

\textbf{Insight 3:} Although \textit{Theorem} \ref{Semantic Consistency of OSDR} provides a standard generalization bound, online trigger reliability under an imperfect oracle is sustained by the closed-loop learning. By employing PPO, the framework implicitly learns to absorb proxy calibration errors and temporal tracking deviations through continuous environmental interaction, as evidenced by UTP results (Sec. VII.C), thereby stabilizing $Q_{sem}(t)$  (Sec. IV.B).


\section{System Model}

\begin{figure}[tbp]
	\centering
	\includegraphics[width=\linewidth]{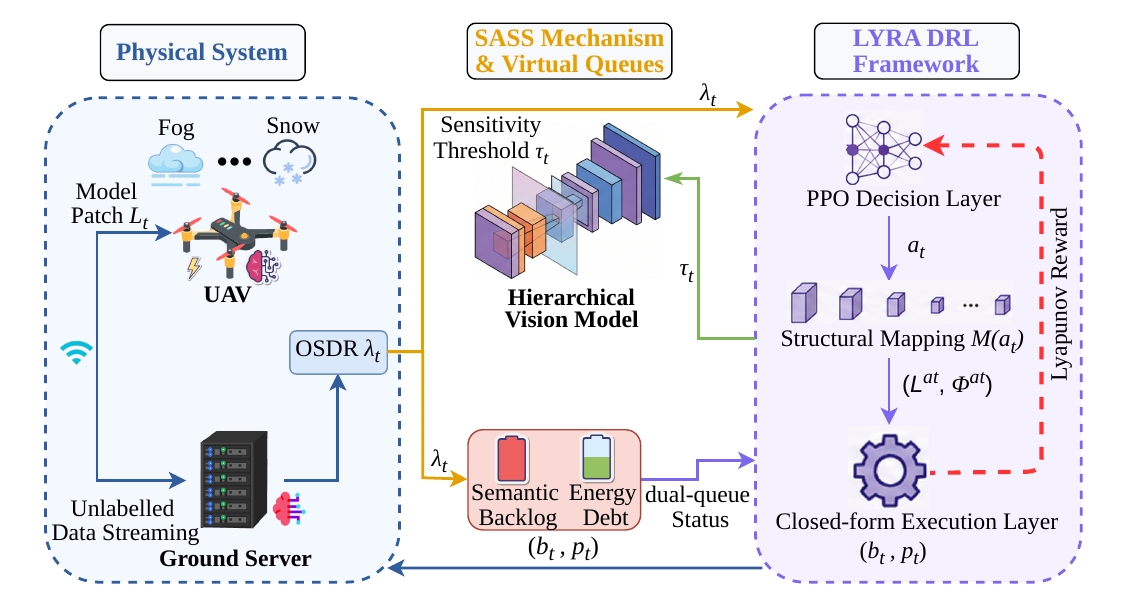}
	\centering
	\caption{The overall architecture of the proposed LYRA framework. It illustrates the real-time data streaming in the physical system, the SASS mechanism driven by the label-free OSDR, and the Lyapunov-guided DRL engine that performs action space dimensionality reduction for optimal resource execution.}
	\label{SystemModel_Framework}
\end{figure}

We consider a MEC-enabled UAV edge system performing mission-critical visual tasks (e.g., real-time target classification). To execute these tasks, the UAV deploys an edge vision model with a hierarchical architecture. The system operates over a finite horizon of discrete time slots $t \in \{0, 1, \dots, T-1\}$. Subjected to low-level environmental corruptions (e.g., dynamic fog or illumination variations) during flight, the UAV must continuously maintain its semantic fidelity by synchronizing its model with a ground-based oracle server. At each slot $t$, the system formulates a joint scheduling and update control vector $\boldsymbol{x}_t \triangleq \{\alpha_t, \tau_t, b_t, p_t\}$, where $\alpha_t \in \{0, 1\}$ is the discrete semantic update trigger, $\tau_t \in [0, 1]$ dictates the layer-wise structural synchronization depth, and $b_t$ and $p_t$ denote the continuous allocated communication bandwidth and server's downlink transmit power, respectively. The overall architecture of the proposed system, bridging the physical operational environment, the virtual queue dynamics, and the LYRA decision framework, is illustrated in Fig. \ref{SystemModel_Framework}.

\subsection{Wireless Communication Model}

We adopt a standard wireless channel model to characterize the link between the UAV and the ground station \cite{r39,r40,r41}. The channel power gain $h_t$ is modeled by large-scale path loss and small-scale fading: $h_t = h_0 d_t^{-\beta} \zeta_t$, where $h_0$ is the reference gain, $d_t$ is the time-varying distance, $\beta$ is the path loss exponent, and $\zeta_t$ denotes small-scale fading. Assume that $N_0$ is the noise power spectral density, according to Shannon's formula, the achievable transmission rate $R_t$ is:
\begin{equation}
    R_t(b_t, p_t) = b_t \log_2 \left( 1 + p_t\cdot h_t/N_0 b_t \right)
\end{equation}

\subsection{Semantic Evolution and Risk Queue}
To quantify VoU under environmental corruption, we define an Accumulated Unserved Semantic Risk Queue $Q_{\text{sem}}(t)$ to track the temporal accumulation of operational risk:
\begin{equation}
    Q_{sem}(t+1) = \max [0, Q_{sem}(t) - \mu_t(\alpha_t, \tau_t) + \lambda_t]
\end{equation}

where $\lambda_t \triangleq \operatorname{OSDR}(t)$ represents the instantaneous semantic risk rate in slot $t$. Maintaining a degraded state without synchronization accumulates decision risk over time, which is eventually served by the structural update reduction $\mu_t$. This reduction $\mu_t$ is governed by the SASS mechanism:
\begin{equation}
    \mu_t(\alpha_t, \tau_t) = \alpha_t \cdot \Phi(\tau_t) \cdot Q_{sem}(t)
    \label{Divergence Reduction}
\end{equation}

Here, $\tau_t \in [0, 1]$ formalizes the layer-wise sensitivity threshold. Since shallow layers processing raw corrupted pixels are most vulnerable to environmental shifts, a high $\tau_t$ restricts updates to the shallowest, most severely degraded neural blocks, while lowering it propagates structural updates deeper. This front-to-back depth dictates the semantic recovery efficiency $\Phi(\tau_t) = 1 - \tau_t^\gamma$ (with shape parameter $\gamma \ge 0$); thus, a smaller $\tau_t$ yields a larger $\Phi(\tau_t)$ for exhaustive divergence elimination. Crucially, Eq. (\ref{Divergence Reduction}) essentially formulates a multiplicative recovery model, where $\Phi(\tau_t)$ characterizes the fractional recovery capability determined by the synchronization depth, while $Q_{sem}(t)$ represents the currently accumulated recoverable semantic divergence. This physically aligns with the proportional catch-up effect in tracking and backlog-draining systems, where patching a more degraded model tends to yield a larger absolute fidelity restoration.

\subsection{Adaptive Patching and Delay Constraint}

Motivated by the structural sparsity observed in Section III.B, we approximate the cumulative patch size using an exponential fitting model:
\begin{equation}
    L_t(\tau_t) = L_{max} \cdot e^{-K\tau_t}
\end{equation}

where $L_{max}$ is the full model size and $K$ is the sparsity coefficient. This approximation captures the rapidly diminishing cumulative parameter volume across hierarchical synchronization depths in practical DNN architectures. To ensure the real-time nature of inspection, the transmission delay $D_t$ must not exceed a maximum threshold $D_{max}$ (which inherently absorbs the fixed probe uplink transmission and server-side oracle inference latency detailed in Sec. III.C):
\begin{equation}
    D_t = \alpha_t\cdot L_t(\tau_t)/R_t(b_t, p_t) \le D_{max}
\end{equation}

\subsection{Energy Consumption and Lyapunov Virtual Queue}


The system total energy consumption $E_{total}(t)$ per slot includes the server's downlink transmission energy and the UAV's baseline operational footprint. Assuming a unit-normalized slot duration:
\begin{equation}
    E_{total}(t) = \alpha_t \cdot p_t \cdot L_t(\tau_t)/R_t(b_t, p_t) + E_{other}
\end{equation}

where $E_{\text{other}}$ represents the baseline energy for UAV computing and sensing including the invariant uplink transmission energy for the highly compressed OSDR probe batch.

To satisfy the long-term energy budget $\mathbb{E}[E_{total}(t)] \le P_{avg}^{budget}$, we introduce a Lyapunov virtual queue $Z(t)$:
\begin{equation}
    Z(t+1) = \max [0, Z(t) + E_{total}(t) - P_{avg}^{budget}]
\end{equation}

$Z(t)$ represents the energy debt. Together with $Q_{sem}(t)$, it constitutes the real-time dual-queue status fed into the LYRA framework (see Fig. \ref{SystemModel_Framework}), effectively guiding the DRL agent to prioritize energy saving when the deficit is high.

\section{Problem Formulation and Lyapunov Optimization}
This section formulates the dynamic joint scheduling and model update process in the MEC system as a stochastic optimization problem. Since $\lambda_t$ and $h_t$ are highly stochastic with unknown prior distributions, we leverage Lyapunov optimization to decouple the long-term time-averaged constraints into a series of deterministic, per-slot optimization sub-problems.

\subsection{Stochastic Optimization Problem Formulation}
The overarching objective of the proposed MEC update and scheduling framework is to minimize the long-term time-averaged system total cost including bandwidth and energy cost while guaranteeing semantic fidelity and energy budgets of the system. Let $\boldsymbol{x}_t \triangleq \{\alpha_t, \tau_t, b_t, p_t\}$ denote the joint control vector. The stochastic optimization problem is formulated as:
\begin{align}
    \textbf{P1}: &\min_{\{\boldsymbol{x}_t\}} \lim_{T \to \infty} \frac{1}{T} \sum_{t=0}^{T-1} \mathbb{E}[\omega_b b_t+\omega_eE_{total}(t)] \nonumber\\ 
    \text{s.t.} \quad &\textbf{C1:} \lim_{T \to \infty} \frac{1}{T} \sum_{t=0}^{T-1} \mathbb{E}[Q_{sem}(t)] < \infty  \nonumber\\
    &\textbf{C2:} \lim_{T \to \infty} \frac{1}{T} \sum_{t=0}^{T-1} \mathbb{E}[E_{total}(t)] \le P_{avg}^{budget}  \nonumber\\
    &\textbf{C3:} \alpha_t\cdot L_t(\tau_t)/R_t(b_t, p_t) \le D_{max}, \quad \forall t  \nonumber\\
    &\textbf{C4:} \quad p_t \in [0, P_{max}], b_t \in [0, B_{max}], \nonumber\\
    &\quad\quad\quad \alpha_t \in \{0, 1\}, \tau_t \in [0, 1]
\end{align}

where $\omega_b$ is the bandwidth weight. C1 imposes strong stability on the accumulated semantic risk queue, ensuring long-term operational risk remains bounded against drift. C2 enforces the long-term energy sustainability of the system. C3 dictates the hard instantaneous delay constraint for critical vision tasks, and C4 specifies physical boundary constraints.

\subsection{Lyapunov Optimization Formulation and Problem Analysis}

Problem $\textbf{P1}$ is a MINLP problem coupled across time, which is generally intractable. To resolve this, we define the joint system state vector as $\boldsymbol{\Theta}(t) \triangleq [Q_{sem}(t), Z(t)]$, encompassing both semantic and energy deficits. The Lyapunov function is constructed as $L(\boldsymbol{\Theta}(t)) = \frac{1}{2}Q_{sem}(t)^2 + \frac{1}{2}Z(t)^2$ to scalarize the system congestion.

The one-step conditional Lyapunov drift is defined as $\Delta(\boldsymbol{\Theta}(t)) = \mathbb{E}[L(\boldsymbol{\Theta}(t+1)) - L(\boldsymbol{\Theta}(t)) | \boldsymbol{\Theta}(t)]$. By incorporating the objective function into the drift, we aim to minimize the Drift-plus-Penalty bound:
\begin{equation}
    \Delta(\boldsymbol{\Theta}(t)) + V \cdot \mathbb{E}[\omega_b b_t +\omega_eE_{total}(t) | \boldsymbol{\Theta}(t)]
\end{equation}

where $V > 0$ is a control parameter balancing the cost minimization and queue stability.

By Eq. (2) and (7), we can establish a rigorous upper bound for (9). Extracting the terms strictly governed by the control variable $\boldsymbol{x}_t$ and dropping the uncontrollable constants, the minimization of the upper bound is mathematically equivalent to solving the following deterministic, per-slot optimization problem ($\textbf{P2}$) at each time slot $t$:
\begin{gather}
    \min_{\boldsymbol{x}_t} \left\{ V \omega_b b_t - Q_{sem}(t) \mu_t(\alpha_t, \tau_t) + (V\omega_e + Z(t)) E_{total}(t) \right\} \nonumber \\
    \text{s.t.} \quad \textbf{C3}, \textbf{C4}
\end{gather}

\begin{theorem}
    Suppose the environmental drift rate $\lambda_t$ and channel gain $h_t$ are i.i.d. with bounded support. Assume the risk backlog operates within a truncated state space $Q_{sem}(t) \le Q_{sem}^{max}$, reflecting the maximum tolerable semantic divergence in mission-critical UAVs. Consequently, the multiplicative drain structure (Eq. 3) guarantees a bounded single-slot divergence reduction within this regime, yielding a finite Lyapunov constant $B$. Assuming exact optimal per-slot minimization of $\textbf{P2}$, the time-averaged system total cost satisfies:$\lim\limits_{T\to\infty} \frac{1}{T} \sum_{t=0}^{T-1} \mathbb{E}[Cost(t)] \le Cost^* + \frac{B}{V}$. While this bound serves as a theoretical benchmark for independent single-slot optimization, LYRA leverages this mathematical structure to sustain long-horizon trajectory performance under dynamic environmental variations.
\end{theorem}

\textbf{\textit{Proof}}: The proof follows standard Lyapunov drift analysis \cite{r5} and is omitted for brevity.

By transforming $\textbf{P1}$ into $\textbf{P2}$, the intractable long-term time-averaged constraints are shifted into instantaneous queue-dependent penalties. To handle the continuous coupling of hybrid decision variables while capturing the temporal correlation of non-stationary environmental drifts, Section VI integrates this Lyapunov drift-plus-penalty structure into a long-horizon DRL framework, utilizing the per-slot penalty to guide a stable long-term reward formulation. 

\section{Lyapunov-Guided discrete DRL}

The deterministic sub-problem $\textbf{P2}$ is a non-convex MINLP, involving coupled discrete $\alpha_t$, continuous $\tau_t$, $p_t$ and $b_t$. Directly solving $\textbf{P2}$ using conventional DRL suffers from dimensionality disaster and poor convergence in hybrid action spaces. To address this, we propose a \textit{Lyapunov-Guided Hierarchical Joint Update Scheduling and Resource Allocation} (LYRA) framework, as detailed in the rightmost module of Fig. \ref{SystemModel_Framework}, consisting of a high-level Decision Layer (PPO) and a low-level Execution Layer (Closed-form Resource Allocation).

\subsection{Action Space Dimensionality Reduction via Structural Mapping}

To circumvent the significant convergence instability inherent in hybrid action spaces, we exploit the block-wise physical topology of DNNs (e.g., ResNet residual blocks, as observed in Section III.B) to design a structural mapping mechanism. Specifically, we collapse the joint continuous-discrete variables, i.e., the update trigger $\alpha_{t}$ and the sensitivity threshold $\tau_{t}$, into a unified, one-dimensional discrete action space $\mathcal{A}=\{0,1,...,K\}$. To bridge the PPO decision and the physical execution depicted in Fig. \ref{SystemModel_Framework}, the structural mapping function $\mathcal{M}(a_{t})$ is defined as:
\begin{equation}
    [\alpha_t, L_t, \Phi_t] = 
    \begin{cases} 
        [0, 0, 0], & \text{if } a_t = 0 \\
        [1, L(\tau^{(a_t)}), \Phi(\tau^{(a_t)})], & \text{if } a_t \in \{1, \dots, K\}
    \end{cases}
\end{equation}

where $\tau^{(a_t)}$ denotes the predefined discrete sensitivity levels corresponding to the hierarchical boundaries of DNN blocks. By discretizing the synchronization depth, this mapping mechanism effectively transforms the combinatorial search for $\{\alpha_{t},\tau_{t}\}$ into a lightweight selection of update granularities.

\begin{lemma}
    Assume there exist finite constants $C_{L}>0$ and $C_{\Phi}>0$ such that $|L(\tau_1) - L(\tau_2)|\leq C_L|\tau_1-\tau_2|$ and $|\Phi(\tau_1) - \Phi(\tau_2)|\leq C_{\Phi}|\tau_1-\tau_2|$ for $\forall \tau_1, \tau_2 \in[0,1]$. Let $\tau^*$ and $\hat{\tau}$ denote the optimal continuous synchronization threshold and its nearest discrete structural level induced by $M(a_t)$ respectively, they satisfy $|\tau^* - \hat{\tau}|\leq \Delta$. Then, $|J(\tau^*) - J(\hat{\tau})|\leq C\Delta$, where $C>0$ is a finite constant determined by the system parameters.
    
    \label{Per-Slot Structural Approximation Bound}
\end{lemma}

\textbf{\textit{Proof}}: Define the per-slot objective component related to $\tau$: $J(\tau)=-Q_{sem}^2\Phi(\tau)+cL(\tau)$, where $c>0$ aggregates bandwidth-energy coefficients induced by the optimal physical execution at the patch size $L(\tau)$. By Lipschitz continuity, $|J(\tau^*) - J(\hat{\tau})|\leq (Q_{sem}^2C_{\Phi} + cC_L)|\tau^* - \hat{\tau}|\leq C\Delta$, where $C=(Q_{sem}^{max})^2C_{\Phi}+c_{max}C_L$ is a finite constant bounded by the system parameters.

\begin{theorem}
    Given any joint queue state $\Theta(t)$, let $\tau_t^*$ be the continuous optimal threshold for subproblem $P_2$, and $\hat{\tau}_t$ be its nearest discrete level induced by $\mathcal{M}(a_t)$ with maximum spacing $\Delta$. The instantaneous structural optimality gap satisfies: $J_t(\hat{\tau}_t) - J_t(\tau_t^*) \le C\Delta$ where $C > 0$ is the finite constant defined in \textbf{Lemma \ref{Per-Slot Structural Approximation Bound}}.
    
    \label{Global Optimality Gap of Structural Mapping}
\end{theorem}

\textbf{\textit{Proof}}: Since the Execution Layer yields the exact physical resource optimum for any given state $\Theta(t)$, the per-slot objective deviation between $\tau_t^*$ and $\hat{\tau}_t$ is governed strictly by the Lipschitz continuity from \textbf{Lemma \ref{Per-Slot Structural Approximation Bound}}. Applying the maximum discretization boundary spacing $\vert{}\tau_t^* - \hat{\tau}_t\vert{} \le \Delta$ yields the bounded instantaneous optimality gap.

\subsection{The Execution Layer: Closed-form Physical Resource Allocation}

Once the discrete action $a_t$ maps to a deterministic patch size $L^{(a_t)}$, the Execution Layer, acting as the underlying physical solver in Fig. \ref{SystemModel_Framework}, focuses exclusively on optimizing continuous $b_t$ and $p_t$ to minimize the penalty function (10).

To minimize the energy penalty, the optimal transmission strategy must strictly meet, but not exceed, the delay constraint C3 (i.e., $D_t = D_{max}$), which yields the required data rate $R_{req}^{(a_t)} = L^{(a_t)} / D_{max}$. By inverting formula (1), the transmit power is expressed as a closed-form function of bandwidth:
\begin{equation}
    p_t(b_t) = \frac{N_0 b_t}{h_t} \left( 2^{R_{req}^{(a_t)}/{b_t}} - 1 \right)
\end{equation}

Based on Eq. (12) and (10), resource allocation can be reduced to a univariate optimization problem over $b_t \in (0, B_{max}]$:
\begin{equation}
    J_{phys}(b_t|a_t)=V\omega_b b_t+(V\omega_e+Z(t))(p_t(b_t)D_{max}+E_{other})
\end{equation}

Since $J_{phys}(b_t | a_t)$ is strictly convex with respect to $b_t$, $b_t^*$ can be obtained by finding the root of its first-order derivative ($\partial J_{phys} / \partial b_t = 0$) using \textbf{Bisection Search} Algorithm. The $p_t^*$ can be subsequently calculated via Eq. (12).

Crucially, To handle physical infeasibility ($p_t^* > P_{max}$), the Execution Layer enforces $E_{total}^* \to \infty$, which heavily penalizes the agent (Eq. 14) to prevent invalid decisions.

\subsection{MDP Transformation and Lyapunov-Guided Reward}
We reformulate the system as a Markov Decision Process (MDP) where the agent interacts with the environment:

\textit{State Space} ($\mathcal{S}$): $s_t = \{h_t, \lambda_t, Q_{sem}(t), Z(t)\}$, representing channel gains, semantic divergence, and semantic/energy deficits.

\textit{Action Space} ($\mathcal{A}$): $a_t \in \{0, 1, \dots, K\}$ (see Section VI-A).

\textit{Reward Function} ($r_t$): As explicitly depicted by the red feedback loop in Fig. \ref{SystemModel_Framework}, we define the reward as the negative of the Lyapunov penalty function to satisfy long-term constraints:
\begin{equation}
    r_t(s_t,a_t)=Q_{sem}(t)\mu_t^{a_t}-(V\omega_e+Z(t))E_{total}^*(a_t)-V\omega_b b_t^*(a_t)
\end{equation}

\subsection{The Decision Layer: PPO}

At the Decision Layer, we deploy an Actor-Critic agent to maximize the expected cumulative discounted reward $J(\theta) = \mathbb{E}[\sum \gamma_{drl}^t r_t]$. To mitigate training instability caused by the non-stationary channel and drift dynamics in mobile edge networks, we adopt the clipped surrogate objective of PPO:
\begin{equation}
    L^{CLIP}(\theta) = \mathbb{E}_t [ \min( \rho_t(\theta) \hat{A}_t, \text{clip}(\rho_t(\theta), 1-\epsilon, 1+\epsilon) \hat{A}_t ) ]
\end{equation}

where $\rho_t(\theta)$ is the probability ratio between the updated and old policies. Actor and Critic networks are updated iteratively using mini-batches from the replay buffer, employing Generalized Advantage Estimation (GAE) $\hat{A}_t$ to reduce variance.

Crucially, the architectural novelty of LYRA lies in this hierarchical MDP reformulation. The deployed PPO agent leverages the historical filtration $\mathcal{F}_{t-1}$ to optimize  the long-horizon discounted cumulative reward, where the Lyapunov-derived penalty reward dynamically enforces long-term constraint compliance across continuous operational trajectories. 

\subsection{Algorithm Design}

Integrating the hierarchical layers discussed above, the complete workflow of the proposed LYRA framework is summarized in \textbf{Alg.} 1. By decoupling the continuous physical resource optimization from the discrete MDP, the framework ensures compliance with system physical constraints while maintaining stable policy evolution.

\begin{algorithm}[htbp]
	\caption{Lyapunov-Guided Hierarchical Joint Update Scheduling and Resource Allocation Framework (LYRA).}
	\label{Online_alg}
	\begin{algorithmic}[1]
		\REQUIRE ~~\\
		$V$, $D_{max}$, $P_{avg}^{budget}$, $\mathcal{A}$ and $\epsilon$.
		\ENSURE ~~\\
		$a_t$, $b_t^*$ and $p_t^*$.
		\STATE Initialize: Actor network $\pi_\theta$ and Critic network $V_\phi$. Set $Q_{sem}(0) = 0$ and $Z(0) = 0$. Empty replay buffer $\mathcal{D}$.
		\FOR {$t = 0, 1, \dots, T-1$}
			\STATE Observe OSDR $\lambda_t$ and channel gain $h_t$.
			\STATE Construct state $s_t = \{h_t, \lambda_t, Q_{sem}(t), Z(t)\}$.
            \STATE Sample discrete action $a_t \sim \pi_\theta(\cdot | s_t)$ from $\mathcal{A}$.
		      \STATE Map $a_t$ to target patch size $L^{(a_t)}$ via Eq. (11).
            \IF {$a_t == 0$}
                \STATE $b_t^* = 0, p_t^* = 0$.
            \ELSE
                \STATE Define convex physical penalty $J_{phys}(b_t | a_t)$ (Eq. 13).
                \STATE Solve $b_t^*$ via \textbf{Bisection Search} on $\partial J_{phys} / \partial b_t = 0$ within (0, $B_{max}$].
                \STATE Calculate $p_t^*$ via Eq. 12; 
                \IF {$p_t^* > P_{max}$}
                    \STATE Mask the action by setting $E_{total}^* = \infty$;
                \ELSE 
                    \STATE Calculate $E_{total}^*$.
                \ENDIF
		      \ENDIF 
            \STATE Compute instantaneous reward $r_t$ using Eq. (14).
            \STATE Update $Q_{sem}(t+1)$ and $Z(t+1)$ using Eq. (2) and (7).
            \STATE Store tuple $(s_t, a_t, r_t, s_{t+1})$ into buffer $\mathcal{D}$.
            \IF {buffer $\mathcal{D}$ is full}
                \STATE Compute GAE $\hat{A}_t$ and optimize Actor and Critic network parameters $\{\theta, \phi\}$ via \textbf{PPO} gradient ascent with the clipped objective (Eq. 15).
                \STATE Clear buffer $\mathcal{D}$.
            \ENDIF
		\ENDFOR
	\end{algorithmic}
\end{algorithm}

\subsection{Complexity Analysis}

For resource-constrained UAV edge nodes, the online scheduling overhead must remain computationally lightweight to support real-time deployment. Therefore, we analyze the per-slot execution complexity of LYRA. Its online execution mainly consists of two sequential components: 1) forward propagation in the PPO Actor network; 2) one-dimensional convex optimization in the Execution Layer. 

For the Actor network with $M$ fully connected layers, the dominant inference complexity arises from matrix multiplications between adjacent layers, yielding: $O(\sum_{m=1}^{M-1} N_m N_{m+1})$, where $N_m$ denotes the number of neurons in the $m$-th layer. Furthermore, in the Execution Layer, if an update action is triggered, the algorithm executes a \textbf{Bisection Search} over the 1D convex bounded interval $(0, B_{max}]$. To converge to the optimal bandwidth $b_t^*$ within an error tolerance $\hat{\epsilon}$, the required number of iterations is $\lceil \log_2(B_{max}/\hat{\epsilon}) \rceil$, where $B_{max}$ is the maximum search interval and $\hat{\epsilon}$ is the convergence tolerance. 

Accordingly, the per-slot execution complexity of LYRA is: $\mathcal{O}(\sum_{m=1}^{M-1} N_m N_{m+1} + \log_2(B_{max}/\hat{\epsilon}))$, which is computationally lightweight for real-time UAV edge deployment.

\section{Experiments}
\subsection{Simulation Setup}

To model A2G channel fluctuations, UAV mobility is driven by high-mobility 2D traces from the well-known and widely utilized CRAWDAD archive \cite{r43} (536 taxis generating 11M+ GPS records over 25 days in the San Francisco Bay Area), which are constrained to urban road topologies and projected into a 3D flight space at a fixed cruising altitude. A2G channel operates with a path loss exponent $\beta=2.2$, noise density $N_0=-169\text{ dBm/Hz}$, bandwidth $B_{\max}=2\text{ MHz}$, and transmit power $P_{\max}=0.5\text{ W}$ \cite{r41}. To ensure responsiveness and sustainability, let $D_{\max}=2\text{ s}$ and $P_{\text{avg}}^{\text{budget}}=0.5\text{ J}$ \cite{r35}, accommodating probe overhead and oracle latency. We employ ResNet-18 as a representative benchmark for block-structured hierarchical DNNs, mapping its residual topology to a discrete action space where $a_t=0$ indicates the idle state and $a_t>0$ specifies the structural synchronization depth. Since LYRA relies on block-level modularity, this setup is architecture-agnostic and generalizes to broader hierarchical DNNs. All model maintenance and policy evaluations run on an NVIDIA RTX 4060 GPU. Parameter $V$ is configured to $10$ to satisfy the $O(1/V)$ trade-off with weights $\omega_b=\omega_e=1.0$, and all results are averaged over 20 random seeds to eliminate stochastic bias.

We adopt CIFAR-10-C, a standardized and severity-controllable benchmark spanning 15 corruption types (e.g., fog, frost, snow) across severities 1–5, for reproducible evaluation under non-stationary semantic drift. Since LYRA operates on divergence rates rather than raw pixel statistics, its scheduling logic is expected to generalize to complex real-world degradations. By sequentially concatenating varying corruption types, we emulate UAV flight dynamics: gradual drift (progressive fog), abrupt drift (sudden contrast shifts), and recurring drift (periodic shadowing). This data stream induces a time-varying environmental risk rate $\lambda_t$ to stress-test triggering precision and long-term constraint satisfaction. Since oracle is periodically refreshed offline, its deviation $\epsilon_t$ remains bounded; \textit{Theorem} and \textit{Proposition} \ref{Asymptotic Optimality of OSDR-Driven Policy} limit resulting performance gap w.r.t. the idealized label-aware policy.

\subsection{Performance Metrics and Baselines}

To quantify operational efficiency and scheduling intelligence, we track three standard metrics, i.e., time-averaged System Total Cost, average risk backlog $Q_{sem}$, and time-averaged patch size (MB/slot), alongside two specialized metrics: (1) \textbf{Semantic Recovery Efficiency (SRE):} Measuring the semantic fidelity improvement achieved per unit of communication overhead, SRE is defined as $\text{SRE} = \frac{1}{|\mathcal{T}^{+}|} \sum_{t \in \mathcal{T}^{+}}\frac{\Phi^{(a_t)}}{L^{(a_t)}}$, where $\mathcal{T}^{+}$ denotes the set of time slots where updates are activated. A higher SRE indicates a more cost-effective selection of synchronization depth. (2) \textbf{Update Triggering Precision (UTP):} Assessing algorithm's ability to identify semantically critical moments, UTP is defined as the ratio of the average risk backlog at update-triggered slots to that at idle slots. A higher UTP signifies that the algorithm precisely captures environmental corruptions and avoids redundant transmissions. We compare LYRA against the following baselines: (i) \textit{\textbf{Periodic Update (PU) \cite{r5}:}} A content-agnostic policy synchronizing at fixed intervals, equivalent to a deterministic AoI trigger insensitive to actual corruptions; (ii) \textit{\textbf{Mixed-Action DDPG (MA-DDPG) \cite{r38}:}} It shares LYRA's network but optimizes within the original hybrid action space to isolate the structural necessity of our hierarchical decoupling. (iii) \textit{\textbf{Heuristic-Reward PPO (HR-PPO) \cite{r40}:}} It shares LYRA's architecture but excludes the Lyapunov term from its reward, isolating the effect of constraint-aware guidance; (iv) \textit{\textbf{Myopic Update (MU) \cite{r42}:}} A greedy per-slot Lyapunov drift-plus-penalty minimizer lacking LYRA's DRL exploration, reacting only to instantaneous queue states; (v) \textit{\textbf{No Update (NU):}} A static policy with no synchronization, serving as a lower bound;  (vi) \textit{\textbf{OSDR-Threshold (OT):}} A static policy triggering model updates whenever the semantic drift rate exceeds a fixed threshold, requiring no learning; (vii) \textit{\textbf{Tuned OSDR Threshold (TOT):}} OT with its threshold tuned offline via search to minimize system total cost. Notably, TOT is used for in-distribution evaluation (Sec. VII.C.3); OT is used for OOD tests (Sec. VII.C.4–5) to avoid leaking unseen drift information via re-tuning.

Furthermore, to conduct a structural ablation on the SASS mechanism, we introduce four synchronization variants operating under identical update triggers: \textit{\textbf{Back-to-Front (B2F)}} (updating parameters backwards from semantic heads), \textit{\textbf{Head Only (HO)}} (updating only the final classifier), \textit{\textbf{Full-Model (FM)}}, and \textit{\textbf{Random Block (RB)}} (updating a randomly selected neural block).

\subsection{Experimental Results and Discussion}

\begin{figure}[tbp]
    \centering
    \subfloat[Reward vs MA-DDPG.]{\includegraphics[width=0.333\columnwidth]{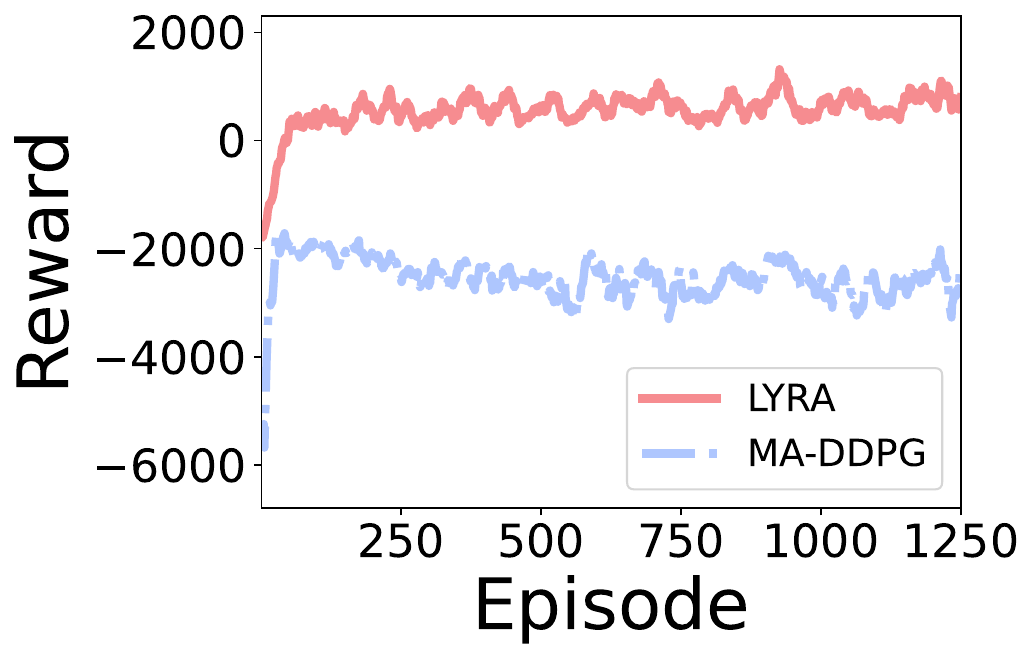}}
    \hfil 
    \subfloat[Reward vs HR-PPO.]{\includegraphics[width=0.333\columnwidth]{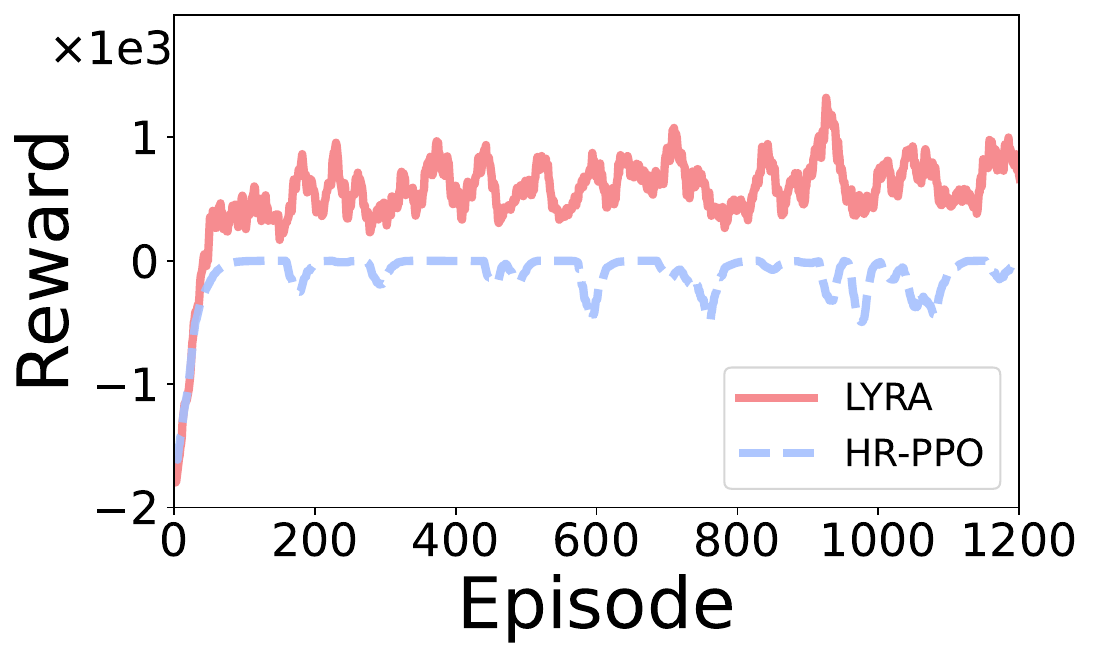}}
    \hfil 
    \subfloat[Energy backlog $Z(t)$.]{\includegraphics[width=0.333\columnwidth]{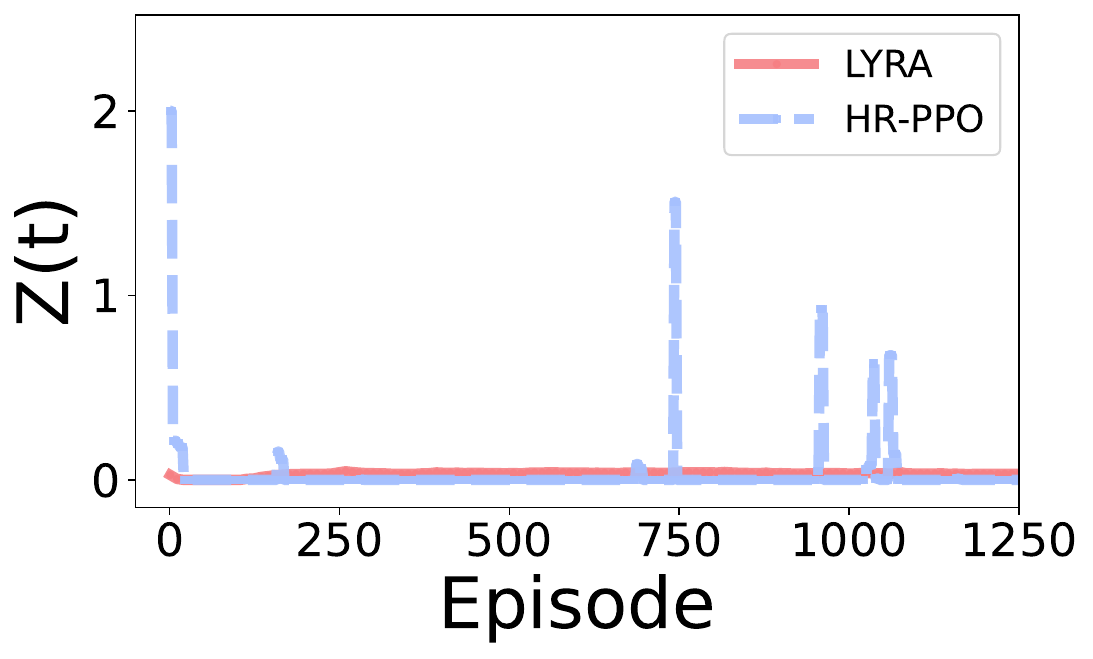}}
    
    \caption{Impact of (a) hierarchical action decoupling against monolithic exploration and (b) Lyapunov guidance on reward convergence, alongside (c) energy queue backlog dynamics.}
    \label{Combined_Ablations}
\end{figure}

\subsubsection{Algorithmic Validation and Constraint Compliance}
To evaluate the effectiveness of LYRA core algorithmic designs, action space dimensionality reduction and Lyapunov-guided reward, we compare its training stability and energy compliance against MA-DDPG and HR-PPO.

As shown in Fig. \ref{Combined_Ablations}(a), LYRA achieves rapid and stable convergence, reaching a high positive reward plateau within the first 100 episodes. Conversely, MA-DDPG fails to converge due to severe gradient conflicts in unreduced hybrid space, proving that conventional monolithic end-to-end exploration struggles under complex MINLP, whereas LYRA's superiority stems from its hierarchical architecture. Furthermore, Fig. \ref{Combined_Ablations}(b) demonstrates LYRA superior reward convergence over HR-PPO, which suffers from significant instability. Crucially, Fig. \ref{Combined_Ablations}(c) reveals the underlying mechanism: LYRA's $Z(t)$ remains bounded, whereas HR-PPO exhibits multiple sharp spikes, indicating severe energy budget violations. This confirms that integrating Lyapunov drift-plus-penalty is imperative for translating long-term hardware constraints into realizable policies.

\begin{table}[t]
    \centering
    \caption{Ablation Study on SASS Mechanism (mean $\pm$ std)}
    \label{SASS_Ablation}
    \begin{tabular}{cccc} 
        \toprule
        \textbf{Algorithm} & $\bm{Q_{sem}}$ & \textbf{Patch Size (MB/slot)} & \textbf{SRE (1/MB)} \\
        \midrule
        LYRA & $2.516 \pm 0.305$ & $0.0459 \pm 0.0040$ & $1.419 \pm 0.217$ \\
        B2F  & $3.098 \pm 0.194$ & $0.2351 \pm 0.1185$ & $0.157 \pm 0.061$ \\
        HO   & $5.641 \pm 0.311$ & $0.0016 \pm 0.0001$ & $21.800 \pm 0.000$ \\
        FM   & $2.936 \pm 0.068$ & $0.1615 \pm 0.0066$ & $0.200 \pm 0.000$ \\
        RB   & $2.812 \pm 0.147$ & $0.0993 \pm 0.0181$ & $1.439 \pm 0.248$ \\
        \bottomrule
    \end{tabular}

    \vspace{4pt} 
    \raggedright 
    \footnotesize 
    \textit{Note:} HO/FM select a fixed block regardless of trigger timing, so their SRE is deterministic (std $\approx$ 0).
\end{table}

\subsubsection{Ablation Study on SASS Mechanism}
To empirically validate the physical necessity of front-to-back synchronization under low-level corruptions, we conduct a structural ablation using fixed scheduling triggers. As seen in Table \ref{SASS_Ablation}, replacing LYRA front-to-back mechanism with B2F updates leads to a significant performance degradation. Notably, SASS reduces the average per-slot communication overhead by 80.5\% compared to B2F (0.0459 MB vs. 0.2351 MB), while achieving a lower risk backlog (2.516 vs. 3.098). This significant inefficiency proves that deep-layer refinements cannot computationally recover the spatial features discarded by corrupted shallow representations. Furthermore, while HO exhibits an inflated SRE due to its negligible patch size, it completely fails to suppress the risk backlog (5.641). By dynamically scaling the front-to-back depth, LYRA achieves the lowest risk backlog while incurring a substantially smaller footprint than all non-degenerate baselines (B2F, FM, RB), best capturing the hierarchical vulnerability of vision models.

\begin{figure*}[tbp]
    \centering
    \subfloat[Recovered Accuracy Dynamics.]{\includegraphics[width=0.5\columnwidth]{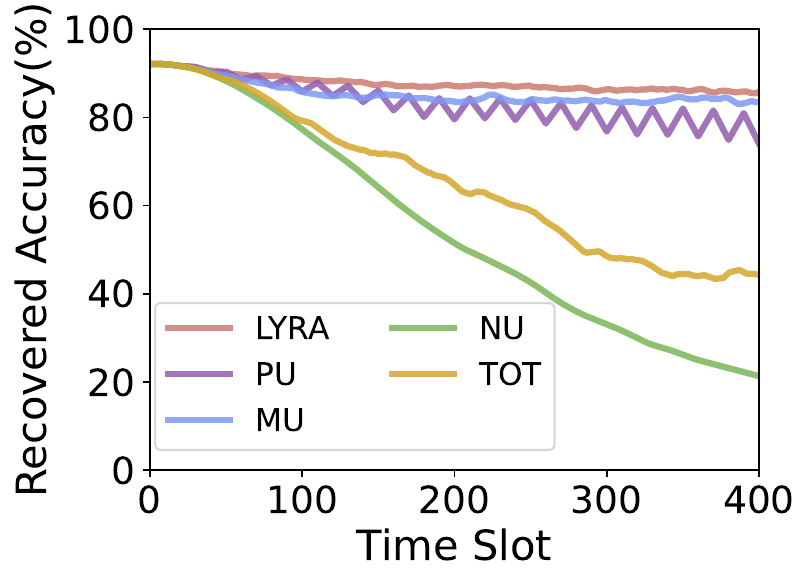}}
    \hfil 
    \subfloat[Semantic Recovery Efficiency.]{\includegraphics[width=0.5\columnwidth]{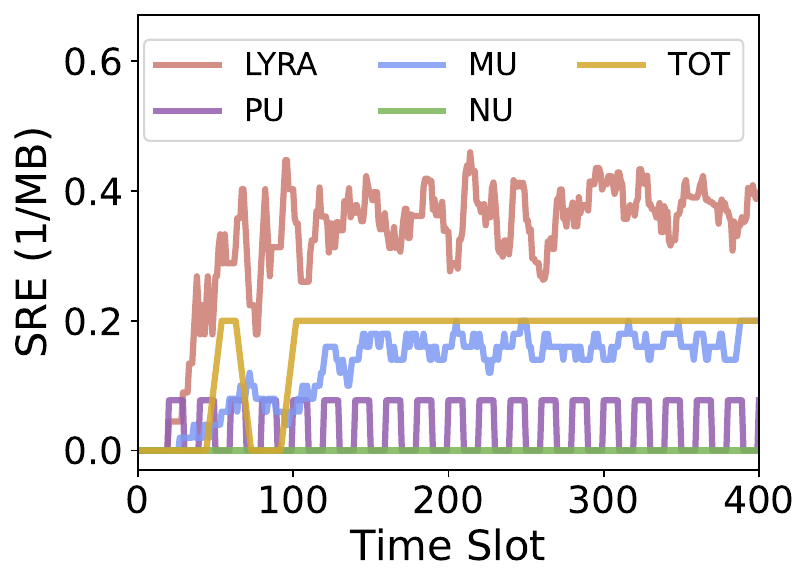}}
    \hfil 
    \subfloat[Update Trigger Precision.]{\includegraphics[width=0.5\columnwidth]{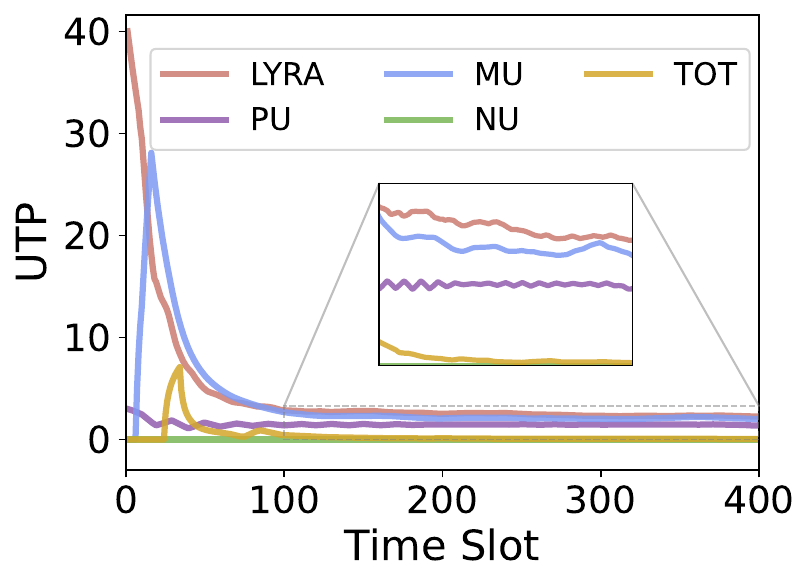}}
    \hfil 
    \subfloat[Cost-Divergence Trade-off.]{\includegraphics[width=0.5\columnwidth]{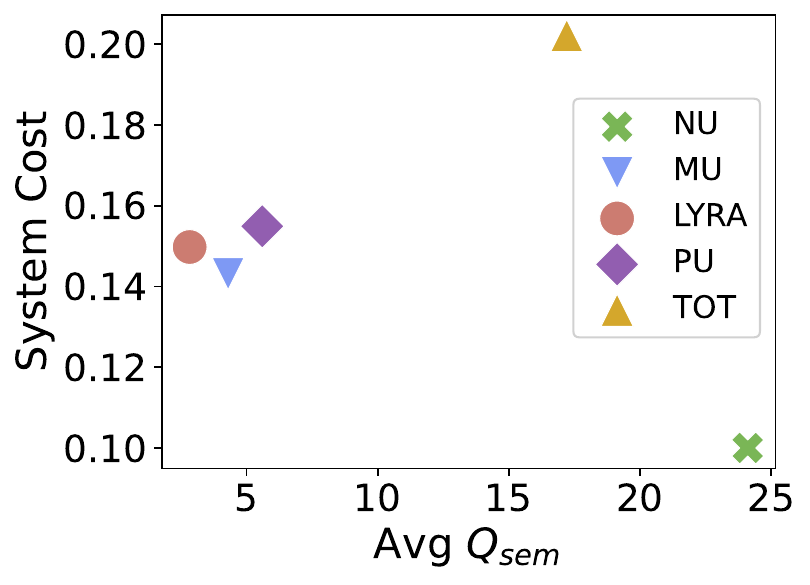}}
    
    \caption{System performance comparison under dynamic environmental corruptions.}
    \label{System_Performance}
\end{figure*}

\subsubsection{System-Level Performance and Trade-off Analysis}
To evaluate the overall performance of LYRA, we compared it against algorithms with different update scheduling strategies.

Fig. \ref{System_Performance}(a) reports the recovered accuracy. LYRA sustains the highest accuracy, closely followed by MU; PU's periodic triggering yields a sawtooth pattern, TOT stabilizes lower, and NU decays to 21\%, mirroring its unbounded backlog growth, confirming that LYRA's backlog control translates into higher recovered accuracy. Fig. \ref{System_Performance}(b) evaluates recovery efficiency: LYRA dynamically modulates patch sizes based on real-time stress, achieving superior SRE, while PU/MU's fixed strategies and TOT's thresholding yield lower per-MB recovery, confirming LYRA's SASS maximizes semantic recovery per MB.  
Fig. \ref{System_Performance}(c) assesses triggering precision. PU wastes bandwidth (UTP $\approx$ 1). Although TOT leverages OSDR, its myopic triggering significantly degrades precision. In contrast, LYRA consistently maintains the highest UTP, accurately targeting semantically critical moments and validating OSDR as a proxy. Ultimately, Fig. \ref{System_Performance}(d) presents the Cost-Divergence trade-off. While NU minimizes cost at the expense of a severe risk backlog, and MU incurs the higher cost for sub-optimal $Q_{sem}$ reduction, TOT reduces risk backlog but achieves highest cost due its lack of state-aware depth control. Conversely, LYRA achieves the most favorable, reducing the average risk backlog by up to 33.3\% over the best baseline at MU-comparable costs. This validates the significance of Lyapunov-guided DRL for forward-looking resource allocation.

\begin{figure}[tbp]
    \centering
    \subfloat[risk backlog.]{\includegraphics[width=0.5\columnwidth]{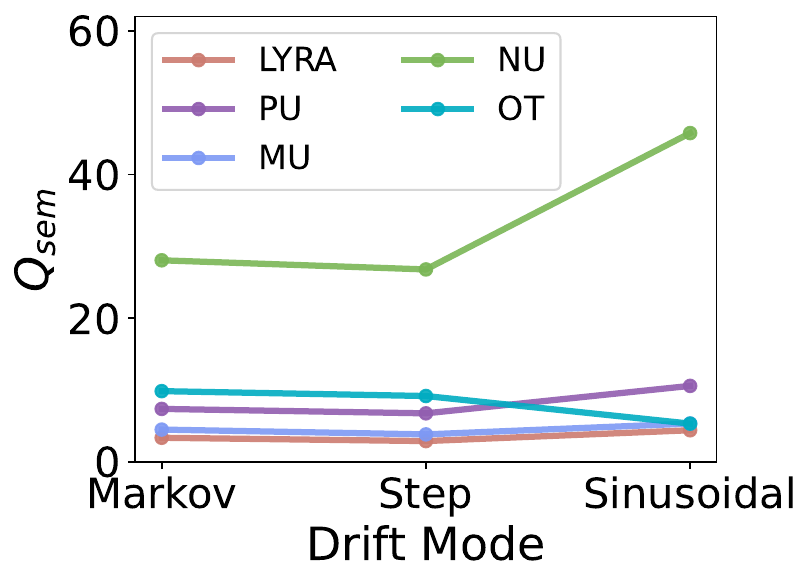}}
    \hfil 
    \subfloat[UTP.]{\includegraphics[width=0.5\columnwidth]{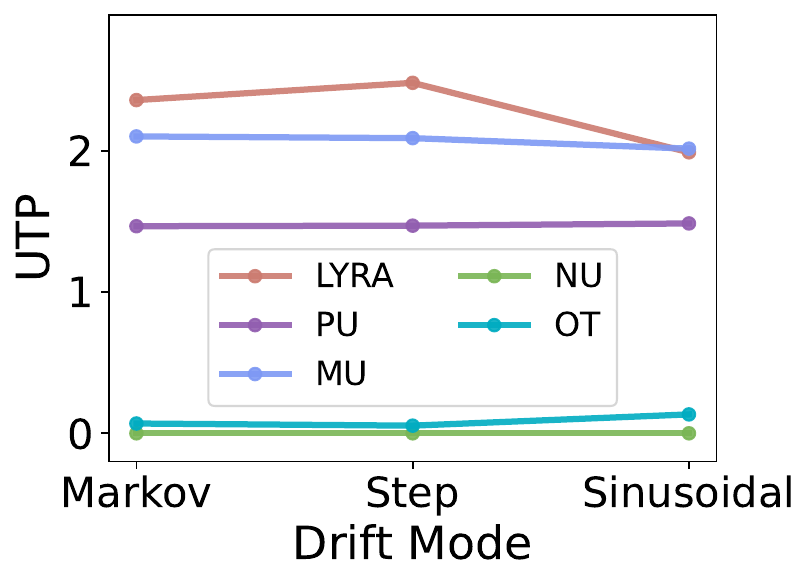}}
    
    \caption{Cross-distribution generalizability comparison under different drift temporal structures.}
    \label{Sim_drift_mode}
\end{figure}

\subsubsection{Generalization Across Unseen Temporal Drift Patterns}
To evaluate zero-shot generalizability across unseen temporal drift structures, we test LYRA under two out-of-distribution (OOD) scenarios: Step (abrupt corruption) and Sinusoidal (periodic fluctuations) modes. In Fig. \ref{Sim_drift_mode}(a), LYRA consistently minimizes the risk backlog, whereas NU suffers a backlog explosion in Sinusoidal mode. Fig. \ref{Sim_drift_mode}(b) shows LYRA achieves peak UTP in Step mode, as the abrupt semantic surge causes a high $\lambda_t$, indicating an immediate, high-value update. In contrast, content-agnostic PU ignores these temporal variations, and the fixed threshold of OT fails under OOD fluctuations. As seen, incorporating real-time divergence $\lambda_t$ into state space is significant for LYRA to avoid overfitting and achieve robust scheduling intelligence across unseen temporal drift dynamics. 

\begin{figure}[tbp]
    \centering
    \subfloat[risk backlog.]{\includegraphics[width=0.333\columnwidth]{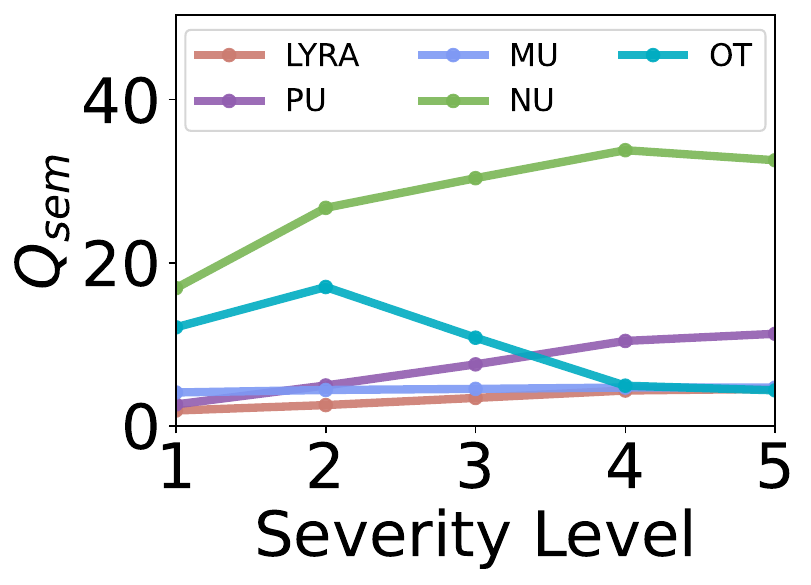}}
    \hfil 
    \subfloat[Patch Size.]{\includegraphics[width=0.333\columnwidth]{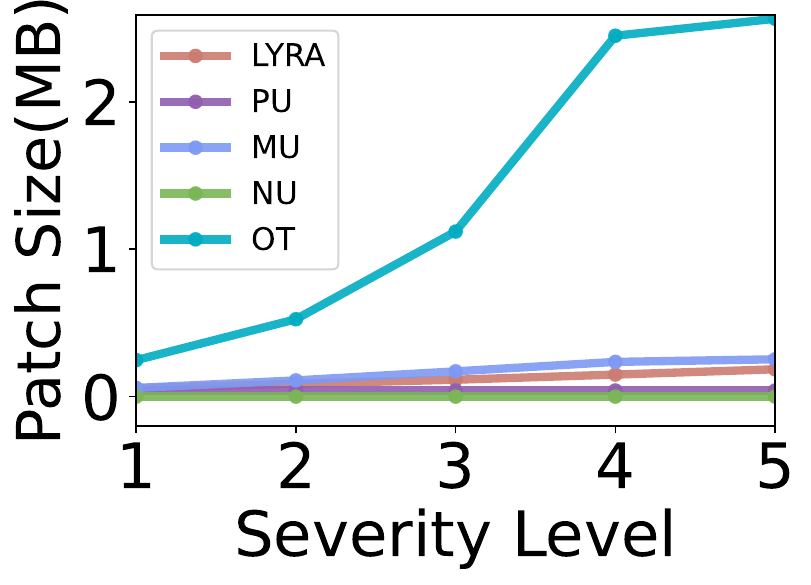}}
    \hfil 
    \subfloat[UTP.]{\includegraphics[width=0.333\columnwidth]{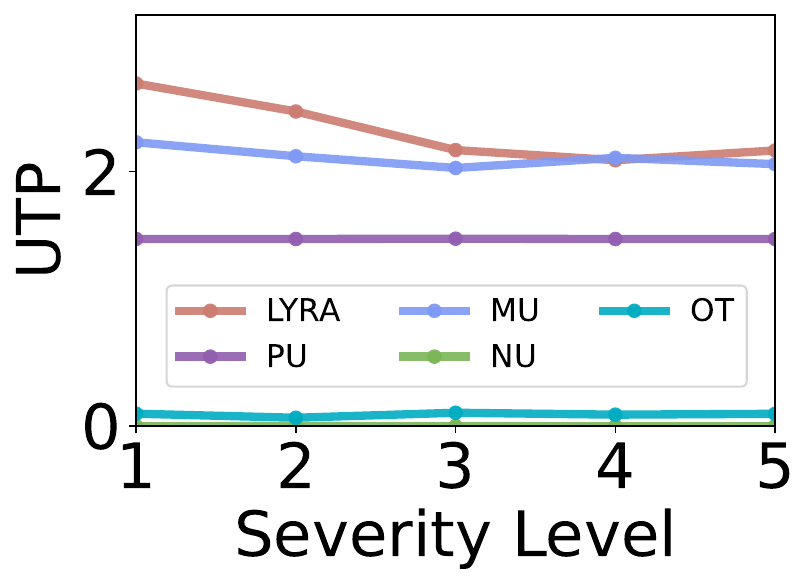}}
    
    \caption{Performance stability comparison under varying drift severity levels.}
    \label{Sim_severity}
\end{figure}

\subsubsection{Performance Stability under Varying Corruption Severities}
This experiment conducts a stress test to evaluate system's scalability and semantic resilience under escalating environmental corruptions. In Fig. \ref{Sim_severity}(a), LYRA consistently maintains the lowest $Q_{sem}$. While OT partially suppresses divergence, its rigid threshold accumulates significant backlog before triggering. Fig. \ref{Sim_severity}(b) demonstrates LYRA’s adaptive depth control: under extreme stress, it autonomously lowers $\tau_t$ for deeper synchronization, preventing cascading error propagation. Conversely, OT indiscriminately triggers updates upon exceeding its static boundary, causing explosive patch size growth. Furthermore, LYRA sustains peak UTP presented in Fig. \ref{Sim_severity}(c), whereas PU stagnates near 1 and OT exhibits poor precision. And, OT's pronounced failure under severe conditions proves that static heuristics cannot substitute state-aware, adaptive scheduling enabled by Lyapunov-guided DRL.

\section{Conclusive Remarks}
This paper originally investigates the joint optimization of model update scheduling and resource allocation for UAV edge intelligence systems under low-level environmental corruption. To ensure the semantic fidelity and resource efficiency of UAVs without relying on expert labels, we propose LYRA, a Lyapunov-guided hierarchical joint update scheduling and resource allocation framework. Specifically, we design SASS to enable bandwidth-efficient and fine-grained structural updates, and formulate OSDR as a label-free proxy for real-time model performance. Furthermore, we explore a DRL that improves decision-making efficiency through action space reduction and satisfies the long-term time-averaged energy constraints. Extensive results demonstrate that LYRA outperforms representative baselines, by achieving superior semantic recovery efficiency and triggering precision, while maintaining resource efficiency. 
As LYRA separates semantic synchronization from specific model architectures, the proposed framework is inherently extensible to a broad range of vision models and edge perception tasks. The obtained results are stimulating our ongoing work to further validate this generality on larger vision models and more diverse UAV perception benchmarks.

\bibliographystyle{IEEEtran}
\bibliography{ref}

\begin{thebibliography}{10}
\providecommand{\url}[1]{#1}
\csname url@samestyle\endcsname
\providecommand{\newblock}{\relax}
\providecommand{\bibinfo}[2]{#2}
\providecommand{\BIBentrySTDinterwordspacing}{\spaceskip=0pt\relax}
\providecommand{\BIBentryALTinterwordstretchfactor}{4}
\providecommand{\BIBentryALTinterwordspacing}{\spaceskip=\fontdimen2\font plus
\BIBentryALTinterwordstretchfactor\fontdimen3\font minus \fontdimen4\font\relax}
\providecommand{\BIBforeignlanguage}[2]{{%
\expandafter\ifx\csname l@#1\endcsname\relax
\typeout{** WARNING: IEEEtran.bst: No hyphenation pattern has been}%
\typeout{** loaded for the language `#1'. Using the pattern for}%
\typeout{** the default language instead.}%
\else
\language=\csname l@#1\endcsname
\fi
#2}}
\providecommand{\BIBdecl}{\relax}
\BIBdecl

\bibitem{r1}
Q.~He, Z.~Dong, F.~Chen \emph{et~al.}, ``Pyramid: Enabling hierarchical neural networks with edge computing,'' in \emph{ACM WWW'22}, New York, NY, USA, 2022, p. 1860–1870.

\bibitem{r2}
H.~Zheng, M.~Gao, Z.~Chen \emph{et~al.}, ``A distributed hierarchical deep computation model for federated learning in edge computing,'' \emph{IEEE Transactions on Industrial Informatics}, vol.~17, no.~12, pp. 7946--7956, 2021.

\bibitem{r3}
A.~Sharshar, L.~U. Khan, W.~Ullah \emph{et~al.}, ``Vision-language models for edge networks: A comprehensive survey,'' \emph{IEEE Internet Things Journal}, vol.~12, no.~16, pp. 32\,701--32\,724, 2025.

\bibitem{r4}
X.~Ren, Q.~Li, H.~Du \emph{et~al.}, ``Tri-ring: Asynchronous service provisioning with online learning in edge cloud networks,'' in \emph{IEEE INFOCOM}, 2025, pp. 1--10.

\bibitem{r5}
X.~Ai and W.~Liang, ``Freshness-aware inference services in edge computing via offloading or local processing,'' in \emph{IEEE LCN}.\hskip 1em plus 0.5em minus 0.4em\relax IEEE, 2025, pp. 1--10.

\bibitem{r6}
K.~Luo, K.~Zhao, T.~Ouyang \emph{et~al.}, ``Efficient coordination of federated learning and inference offloading at the edge: A proactive optimization paradigm,'' \emph{IEEE Transactions Mobile Computing}, vol.~24, no.~1, pp. 407--421, 2025.

\bibitem{r7}
Y.~Zeng, R.~Zhou, L.~Jiao \emph{et~al.}, ``Efficient online dnn inference with continuous learning in edge computing,'' in \emph{IEEE/ACM IWQoS}, 2024, pp. 1--10.

\bibitem{r8}
X.~Ai, W.~Liang, and C.~Liu, ``Joint optimization of model retraining and inference services in dt-assisted edge computing,'' \emph{IEEE/ACM Transactions Networking}, vol.~34, pp. 1804--1819, 2025.

\bibitem{r9}
J.~Li, S.~Guo, W.~Liang \emph{et~al.}, ``Digital twin-enabled service provisioning in edge computing via continual learning,'' \emph{IEEE Transactions Mobile Computing}, vol.~23, no.~6, pp. 7335--7350, 2024.

\bibitem{r10}
Y.~Kong, P.~Yang, and Y.~Cheng, ``Adaptive on-device model update for responsive video analytics in adverse environments,'' \emph{IEEE Transactions on Circuits and Systems for Video Technology}, vol.~35, no.~1, pp. 857--873, 2025.

\bibitem{r11}
H.~Liu, S.~Liu, S.~Long \emph{et~al.}, ``Joint optimization of model deployment for freshness-sensitive task assignment in edge intelligence,'' in \emph{IEEE INFOCOM}, 2024, pp. 1751--1760.

\bibitem{r12}
\BIBentryALTinterwordspacing
P.~Han, S.~Wang, Y.~Jiao \emph{et~al.}, ``Federated learning while providing model as a service: Joint training and inference optimization,'' 2023. [Online]. Available: \url{https://arxiv.org/abs/2312.12863}
\BIBentrySTDinterwordspacing

\bibitem{r13}
J.~Li, J.~Wang, W.~Liang \emph{et~al.}, ``Inference service fidelity maximization in dt-assisted edge computing,'' \emph{IEEE Transactions Mobile Computing}, 2025.

\bibitem{r14}
Y.~Zhang, W.~Liang, Z.~Xu \emph{et~al.}, ``Aoi-aware inference services in edge computing via digital twin network slicing,'' \emph{IEEE Transactions Services Computing}, vol.~17, no.~6, pp. 3154--3170, 2024.

\bibitem{r15}
X.~Ai, W.~Liang, Y.~Zhang \emph{et~al.}, ``Fidelity-aware inference services in dt-assisted edge computing via service model retraining,'' \emph{IEEE Transactions Services Computing}, 2025.

\bibitem{r16}
\BIBentryALTinterwordspacing
H.~Cai, Z.~Zhou, and Q.~Huang, ``Online resource allocation for edge intelligence with colocated model retraining and inference,'' 2024. [Online]. Available: \url{https://arxiv.org/abs/2405.16029}
\BIBentrySTDinterwordspacing

\bibitem{r17}
T.~Li, S.~Leng, X.~Liao \emph{et~al.}, ``Digital twin-based task-driven resource management in intelligent uav swarms,'' \emph{IEEE Transactions on Intelligent Transportation Systems}, vol.~26, no.~4, pp. 5467--5480, 2025.

\bibitem{r18}
W.~Jiang, H.~Han, D.~Feng \emph{et~al.}, ``Energy-efficient and accuracy-aware dnn inference with iot device-edge collaboration,'' \emph{IEEE Transactions on Services Computing}, 2025.

\bibitem{r19}
A.~Mahadevan and M.~Mathioudakis, ``Cost-effective retraining of machine learning models,'' \emph{arXiv preprint arXiv:2310.04216}, 2023.

\bibitem{r20}
T.~Zhang, Z.~Chen, A.~Liu \emph{et~al.}, ``Age of information under periodic updating: Time-dependent statistical characteristics,'' \emph{IEEE Transactions on Communications}, 2026.

\bibitem{r21}
M.~Hatami, M.~Leinonen, and M.~Codreanu, ``Aoi minimization in status update control with energy harvesting sensors,'' \emph{IEEE Transactions on Communications}, vol.~69, no.~12, pp. 8335--8351, 2021.

\bibitem{r22}
G.~Zhang, C.~Shen, Q.~Shi \emph{et~al.}, ``Aoi minimization for wsn data collection with periodic updating scheme,'' \emph{IEEE Transactions on Wireless Communications}, vol.~22, no.~1, pp. 32--46, 2022.

\bibitem{r23}
J.~Qi, C.~Liu, C.~Xu \emph{et~al.}, ``Efficient information updates in compute-first networking via reinforcement learning with joint aoi and voi,'' \emph{IEEE Internet of Things Journal}, 2026.

\bibitem{r24}
H.~Lee, S.~Yoo, D.~Lee \emph{et~al.}, ``How important is periodic model update in recommender system?'' in \emph{Proceedings of the 46th International ACM SIGIR Conference on Research and Development in Information Retrieval}, 2023, pp. 2661--2668.

\bibitem{r25}
S.~Yao, N.~Rashvand, A.~D. Pazho \emph{et~al.}, ``From offline to periodic adaptation for pose-based shoplifting detection in real-world retail security,'' \emph{IEEE Internet of Things Journal}, 2026.

\bibitem{r26}
Q.~Zhang, R.~Han, C.~H. Liu \emph{et~al.}, ``Edgeta: Neuron-grained scaling of foundation models in edge-side retraining,'' \emph{IEEE Transactions on Mobile Computing}, vol.~24, no.~4, pp. 2690--2707, 2024.

\bibitem{r27}
J.~Liu, R.~Li, H.~Xu \emph{et~al.}, ``Fedquad: Adaptive layer-wise lora deployment and activation quantization for federated fine-tuning,'' \emph{IEEE Transactions on Mobile Computing}, 2025.

\bibitem{r28}
L.~Xu, J.~Jiao, T.~Yang \emph{et~al.}, ``Semantic utility loss of information for energy efficient semantic status update communications,'' \emph{IEEE Transactions on Cognitive Communications and Networking}, vol.~11, no.~1, pp. 59--74, 2024.

\bibitem{r29}
K.~Huang, Q.~Lan, Z.~Liu \emph{et~al.}, ``Semantic data sourcing for 6g edge intelligence,'' \emph{IEEE Communications Magazine}, vol.~61, no.~12, pp. 70--76, 2023.

\bibitem{r30}
Q.~Hu, G.~Zhang, Z.~Qin \emph{et~al.}, ``Robust semantic communications with masked vq-vae enabled codebook,'' \emph{IEEE Transactions on Wireless Communications}, vol.~22, no.~12, pp. 8707--8722, 2023.

\bibitem{r31}
L.~Xia, Y.~Sun, X.~Li \emph{et~al.}, ``Wireless resource management in intelligent semantic communication networks,'' in \emph{IEEE INFOCOM 2022-IEEE Conference on Computer Communications Workshops (INFOCOM WKSHPS)}.\hskip 1em plus 0.5em minus 0.4em\relax IEEE, 2022, pp. 1--6.

\bibitem{r32}
Z.~Yang, M.~Chen, Z.~Zhang \emph{et~al.}, ``Performance optimization of energy efficient semantic communications over wireless networks,'' in \emph{2022 IEEE 96th Vehicular Technology Conference (VTC2022-Fall)}.\hskip 1em plus 0.5em minus 0.4em\relax IEEE, 2022, pp. 1--5.

\bibitem{r33}
H.~Xie, Z.~Qin, X.~Tao \emph{et~al.}, ``Task-oriented multi-user semantic communications,'' \emph{IEEE Journal on Selected Areas in Communications}, vol.~40, no.~9, pp. 2584--2597, 2022.

\bibitem{r34}
P.~Qin, Y.~Fu, J.~Zhang \emph{et~al.}, ``Drl-based resource allocation and trajectory planning for noma-enabled multi-uav collaborative caching 6g network,'' \emph{IEEE Transactions on Vehicular Technology}, vol.~73, no.~6, pp. 8750--8764, 2024.

\bibitem{r35}
Y.~Chen, Y.~Yang, Y.~Wu \emph{et~al.}, ``Joint trajectory optimization and resource allocation in uav-mec systems: A lyapunov-assisted drl approach,'' \emph{IEEE Transactions on Services Computing}, 2025.

\bibitem{r36}
B.~Yin, X.~Fang, and X.~Wang, ``Joint optimization of trajectory control, resource allocation, and user association based on drl for multi-fixed-wing uav networks,'' \emph{IEEE Transactions on Wireless Communications}, vol.~23, no.~10, pp. 13\,330--13\,343, 2024.

\bibitem{r37}
X.~Wu, L.~Liang, W.~Wen \emph{et~al.}, ``Drl-based trajectory optimization and computation-aware resource allocation for uav-assisted edge computing networks,'' \emph{IEEE Internet of Things Journal}, 2025.

\bibitem{r38}
Z.~Chang, H.~Deng, L.~You \emph{et~al.}, ``Trajectory design and resource allocation for multi-uav networks: Deep reinforcement learning approaches,'' \emph{IEEE Transactions on Network Science and Engineering}, vol.~10, no.~5, pp. 2940--2951, 2022.

\bibitem{r39}
Y.~Qin, Z.~Zhang, X.~Li \emph{et~al.}, ``Deep reinforcement learning based resource allocation and trajectory planning in integrated sensing and communications uav network,'' \emph{IEEE Transactions on Wireless Communications}, vol.~22, no.~11, pp. 8158--8169, 2023.

\bibitem{r40}
C.~Zhang, Z.~Li, C.~He \emph{et~al.}, ``Deep reinforcement learning based trajectory design and resource allocation for uav-assisted communications,'' \emph{IEEE Communications Letters}, vol.~27, no.~9, pp. 2398--2402, 2023.

\bibitem{r41}
H.~Wang, H.~Zhang, X.~Liu \emph{et~al.}, ``Joint uav placement optimization, resource allocation, and computation offloading for thz band: A drl approach,'' \emph{IEEE Transactions on Wireless Communications}, vol.~22, no.~7, pp. 4890--4900, 2022.

\bibitem{r42}
J.~Lin, L.~Huang, H.~Zhang, X.~Yang, and P.~Zhao, ``A novel lyapunov based dynamic resource allocation for uavs-assisted edge computing,'' \emph{Computer Networks}, vol. 205, p. 108710, 2022.

\bibitem{r43}
\BIBentryALTinterwordspacing
M.~Piorkowski, N.~Sarafijanovic-Djukic, and M.~Grossglauser, ``Crawdad data set epfl/mobility (v. 2009-02-24),'' Feb. 2009. [Online]. Available: \url{https://crawdad.org/epfl/mobility/20090224}
\BIBentrySTDinterwordspacing

\end{thebibliography}

\end{document}